\documentclass[english,aps,prd,reprint,superscriptaddress,floatfix,amsfonts,nofootinbib]{revtex4-1}
\usepackage{amssymb}
\usepackage{graphicx}
\usepackage{dsfont}
\usepackage{amsmath}
\usepackage{bm}
\usepackage{times}
\usepackage{color}
\usepackage[applemac]{inputenc}
\usepackage[nice]{nicefrac}
\usepackage[colorinlistoftodos]{todonotes}
\usepackage{ esint }
\usepackage{xr}
\usepackage{upgreek}
\makeatother
\usepackage{amsthm}
\newtheorem{thm}{Theorem}

\PassOptionsToPackage{normalem}{ulem}
\usepackage{ulem}

\graphicspath{{converted_graphics/}}

\begin{document}

\title{Exact Solution for the Heat Conductance in Harmonic Chains}

\begin{abstract}
We present an exact solution for the heat conductance along a harmonic chain connecting two reservoirs at different temperatures. In this model, the end points correspond to Brownian particles with different damping coefficients. Such analytical expression for the heat conductance covers its behavior from  mesoscopic to very long one-dimensional quantum chains, and validates the ballistic nature of the heat transport in the latter example. This implies the absence of the Fourier law for classical and quantum harmonic chains. We  also provide a thorough analysis of the normal modes of system which helps us to satisfactorily interpret these results.
\end{abstract}

\author{G. A. Weiderpass}
\altaffiliation[Present address: ]{Department of Physics, The University of Chicago, Chicago, Illinois 60637, USA}
\affiliation{Instituto de  F\'\i sica Gleb Wataghin, Universidade Estadual de Campinas, 13083-859, Campinas, SP, Brazil}
\author{Gustavo M. Monteiro}
\altaffiliation[Present address: ]{Department of Physics, City College, City University of New York, New York, NY 10031, USA}
\affiliation{Instituto de  F\'\i sica Gleb Wataghin, Universidade Estadual de Campinas, 13083-859, Campinas, SP, Brazil}
\author{A. O. Caldeira}
\affiliation{Instituto de  F\'\i sica Gleb Wataghin, Universidade Estadual de Campinas, 13083-859, Campinas, SP, Brazil}

\maketitle

\section{Introduction\label{intro}}

Fourier law in classical macroscopic systems is a well-tested phenomenological statement for both liquids and solids. It states that the heat flux density $\boldsymbol{J_E}$, which is the amount of heat that flows through a unit area per unit time, can be locally expressed as
\begin{align}
    \boldsymbol{J_E}=-\kappa \boldsymbol \nabla T,
\end{align}
where $\kappa$ is the thermal conductivity and $T$ the local temperature.  This dependence on the temperature gradient implies that if we connect two thermal reservoirs by a bar with thermal conductivity $\kappa$, the heat flux will decrease with the length of the bar. However, since the work of Rieder \textit{et al.} \cite{Rieder}, it has been shown that modeling the solid connecting the reservoirs as a harmonic chain leads to a heat flux which is independent of the length of the bar. This implies that the thermal  conductance does not depend on the number of oscillators, which is a characteristic of ballistic transport. This behavior is known as anomalous heat conduction. Rieder analysis however takes only classical effects into account and relies on approximations on the coupling mechanism with the reservoirs, white noise, and Gaussian statistics. An anomalous temperature profile was also found along the bar. The temperature of the bar is constant until we reach oscillators very near its endpoints  where it starts decreasing (increasing) and then sharply increases (decreases) to reach the temperature of the hot (cold) reservoir \cite{Rieder}. As pointed out by \cite{DharSpohn,Lebowitz}, this result is believed to be a feature of integrable systems.

A large effort was made to take into account quantum effects, different coupling mechanisms, mass disorder, and go beyond the Gaussian approximation \cite{DharSpohn,Lebowitz,Das,Spohn,Das2,O'Connor,Zurcher,Saito,Segal,Dhar,Bonetto,BetiniReview,Panasyuk,Dhar2,Asadian,SaitoDhar,DharReview}. Of particular importance is the work by Segal \textit{et al.} \cite{Segal} where numerical studies using Restricted Hartree-Fock methods have hinted an $1/N$ dependence on the heat flux. However, the exact solution of a linear chain coupled with two distinct reservoirs taking into account, coloured noise, and non-Gaussian statistics remains elusive \cite{Dhar2,Asadian}. For a review on the subject of transport in one-dimensional systems, we refer to \cite{DharReview,BetiniReview,Bonetto}.

A second approach to the problem of heat conduction through mesoscopic devices is based on Landauer theory \cite{Landauer1,Landauer2,Glazman,Kirczenow,Szafer,Haanappel,Escapa,Tekman,Avishai,Castano}. Changing the formalism of Landauer theory to phonons, an expression for the heat conduction through a molecular junction was derived \cite{Angelescu,Rego,Blencowe}. Also based on the quantization of the electrical conductance \cite{Wees,Wharam}, a quantum thermal conductance was proposed \cite{Angelescu,Rego,Blencowe} and measured experimentally  soon afterwards \cite{Schwab}.

A third approach which has gained a lot of momentum in the last decade uses the formalism of dynamical semigroups \cite{Breuer}, to study 1-dimensional chains described by the spin $1/2$ XXZ and Fermi-Hubbard models \cite{Znidaric1,Znidaric2,Mendonza,Prosen1,Prosen2}. For a review on this approach we refer to \cite{BetiniReview}.

In this contribution, we study the heat transport in mesoscopic systems for both quantum and classical regimes following the first approach, that is, using quantum Brownian motion formalism for open quantum system \cite{Breuer,Weiss,caldeira1}. For such, we consider a finite harmonic chain  whose endpoints are coupled to thermal reservoirs, held at different temperatures. In sec. II we present the model we use throughout the paper, and establish the expression which describes the exchange of heat between the thermal reservoirs via the harmonic chain, namely, the heat flux. In sec. III we propose the method by which we find an analytical solution for the thermal conductance in terms of an integral which can be solved numerically. In sec IV we study the normal modes of our system when the coupling constants are the same on both ends of the chain. The discussion of different coupling constant is presented in the supplemental material. In sec V we show how to pertubativelly obtain an expression of the normal modes in the limit of large number of particles in the chain. In sec VI we present the closed analytical expression of the heat flux through the harmonic chain for both quantum and classical underling mechanics, and we discuss the thermal conductance in the appropriate limit on which this quantity can be defined. Finally, we present our conclusions in sec. VII.
 
\section{The model\label{model}}

 The particular set up we use to describe our system is  schematically shown in Fig.~\ref{chain}. It is composed of a finite chain of $N$ particles of mass $m$, each of which is coupled to its nearest neighbors by harmonic springs  with frequency $\omega_0$. Besides, the  endpoint particles are also coupled to two distinct thermal baths, $L$ and $R$, held at temperatures $T_{L}$ and  $T_{R}$, respectively. We then assume that the coupling between the endpoint particles and the environments can be cast into the well-known form of a particle coupled to a bath of non-interacting harmonic oscillators \cite{caldeira1,Breuer,Weiss} which, with a specific choice of the so-called spectral function (see below), has been used in the literature to describe quantum Brownian motion. Therefore, the  system's dynamics is governed by the Hamiltonian
\begin{align}\label{H}
 H&=\sum\limits_{j=1}^N \,\frac{P_j^2}{2m}+\sum\limits_{j=1}^{N-1}\frac{m\omega_0^2}{2}\left(X_{j+1}-X_j\right)^2 +\nonumber
 \\ 
&\sum\limits_{\substack{a=L,R\\ i}}\left[ \frac{p_{a i}^2}{2m_{ai}}+\frac{m_{ai}\,\omega_{ai}^2}{2}\left(q_{ai}-\frac{C_{ai}\,X_{a}}{m_{ai}\,\omega_{ai}^2}\right)^2\right],
\end{align} 
\noindent where the sets $\{X_j,P_j\}$ and $\{q_{ai},p_{ai}\}$ refer to the canonical coordinates of the  chain and both reservoir ($a=L,R$) particles, respectively.  We have  also identified $X_1:= X_L$ and $X_N:= X_R$ in Eq.~(\ref{H}), in order to shorten the notation. Moreover, $m_{a i}$, $\omega_{ai}$, and $C_{ai}$ are, respectively, the masses, frequencies and coupling constants of each environment oscillator.
\begin{figure} [h]

\centering
\includegraphics[scale=.36]{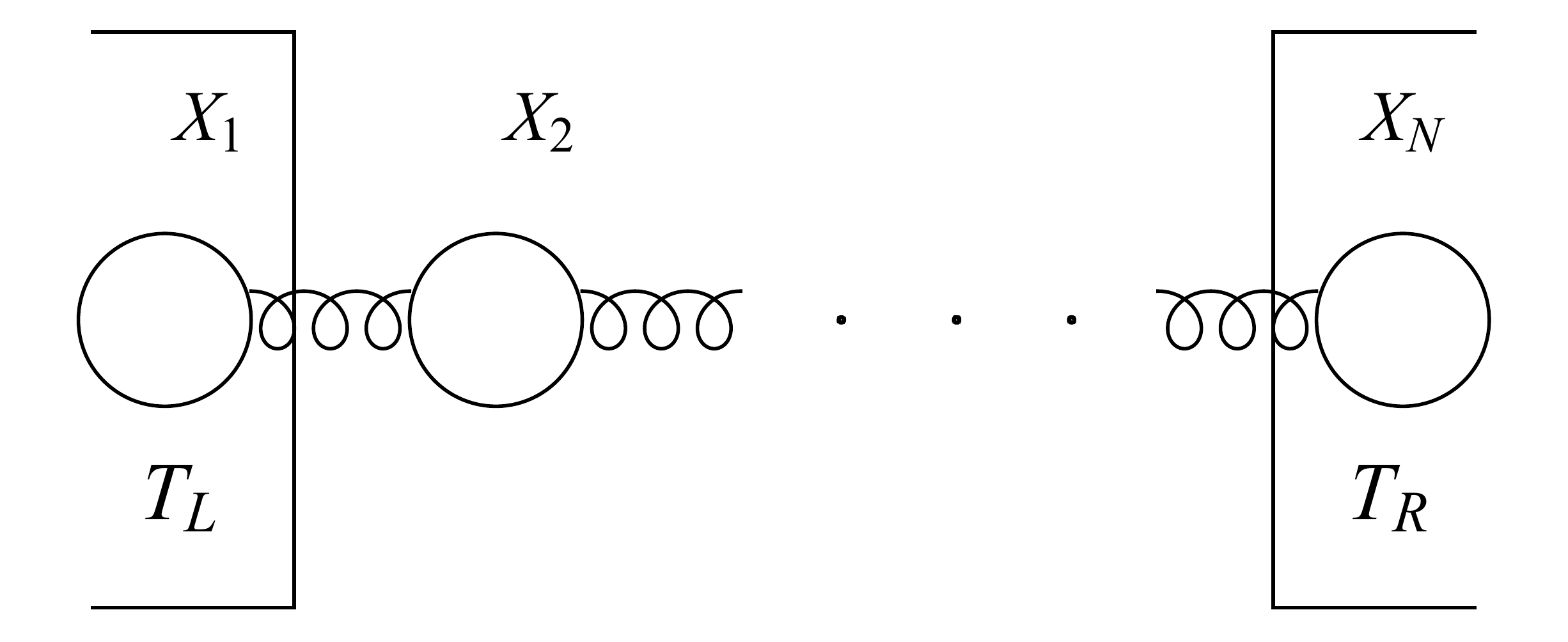}
\caption{Model of an harmonic chain coupled to thermal reservoirs at temperatures $T_L$ and $T_R$. This chain is described by the set of coupled equations (\ref{X1}-\ref{XN}).}
\label{chain}
\end{figure} 

 In order to ensure that the endpoints perform  Brownian motion, we have to model the spectral function, $\mathcal J_{a}(\omega)$, of each reservoir by 
\begin{equation}\label{J}
\mathcal J_{a}(\omega):= \frac{\pi}{2} \sum\limits_{i}\frac{C_{ai}^2}{m_{ai}\omega_{ai}}\delta(\omega-\omega_{ai})=\eta_{a}\omega\,\Theta(\Omega-\omega),
\end{equation}
where $\Theta(\Omega-\omega)$ is the heavyside function,  $\Omega$ is a high frequency cutoff, and $\eta_{a}$ is the damping constant, which shows up in the Langevin equation.

Writing the Heisenberg equations of motion for all the operators involved in (\ref{H}), and following the standard procedure to eliminate the environment coordinates thereof \cite{caldeira1,Breuer,Weiss} we end up with
\begin{align}
    &\ddot{X}_1+\frac{\eta_L}{m}\dot{X}_1-\,\omega_0^2(X_{2}-X_1)=\frac{F_L(t)}{m}, \label{X1}
    \\
     &\ddot{X}_j-\omega_0^2(X_{j+1}-2X_j+X_{j-1})=0,\quad 2\leq j< N, \label{Xj}
   \\ 
    &\ddot{X}_N+\frac{\eta_R}{m}\dot{X}_N+\omega_0^2(X_{N}-X_{N-1})=\frac{F_R(t)}{m}, \label{XN}
\end{align}
in the long time ($\Omega\rightarrow\infty$) approximation. Here, $F_{L,R}(t)$ correspond to fluctuating forces which average to zero, and possesses non-vanishing correlations, that is,
\begin{align}
     &\frac{1}{2}\langle\{F_{a}(t),F_{b}(t')\}\rangle=\nonumber
     \\
     &=\frac{\eta_{a}\hbar\delta_{ab}}{\pi}\int\limits_{0}^{\infty}d\omega\omega\coth\left(\frac{\hbar\omega}{2k_BT_a}\right)\cos(\omega(t'-t))\nonumber
     \\ 
     &=\frac{\eta_{a}\hbar\delta_{ab}}{2\pi}\int\limits_{-\infty}^{\infty}d\omega\omega\coth\left(\frac{\hbar\omega}{2k_BT_a}\right)e^{-i\omega(t'-t)}, \label{Noise}
\end{align}
for $a,b=L,R$. In Eq.~(\ref{Noise}), the angle brackets refer to the thermal average (expectation value over the bath degrees of freedom). 

Similar versions of this model have already been studied using white instead of the coloured noise \cite{Zurcher}, Born-Markov approximation (which fails to capture quantum effects like entanglement \cite{Saito,Dhar2,Asadian}), numerical methods \cite{Segal} or treating the coupling perturbatively \cite{Dhar}. However, we here provide an {\it analytical solution} for the heat flux along a chain connecting two identical reservoirs at different temperatures, in the stationary regime. Such expression  provides a fully 
quantum mechanical description of the heat transfer along the chain. 

The stationary regime of this system is achieved when all the energy transferred from the hotter reservoir to the chain corresponds exactly to the entire energy transferred from the chain to the colder reservoir. In other words, the total energy of the chain must remain constant with time. Since the chain energy is defined through the expectation value of the chain Hamiltonian, that is,
\begin{equation}
    E_{chain}=\left\langle\sum_{j=1}^N\frac{m}{2}\dot X_j^2+\sum_{j=1}^{N-1}\frac{k}{2}(X_{j+1}-X_j)^2\right\rangle,
\end{equation}
it is easy to show, using Eqs.~(\ref{X1}-\ref{XN}), that
\begin{align}
    \frac{d}{dt}E_{chain}=\;&\frac{1}{2}\left\langle\{\dot{X}_1(t),F_L(t)\}\right\rangle-\eta_L\left\langle\dot{X}_1^2(t)\right\rangle\nonumber
    \\
    &+\frac{1}{2}\left\langle\{\dot{X}_N(t),F_R(t)\}\right\rangle-\eta_R\left\langle\dot{X}^2_N(t)\right\rangle\,.\label{energy-rate}
\end{align}
Since the coupling to both reservoirs dampens the chain normal modes, these oscillators are expected to reach the stationary regime for sufficiently long time regardless of their initial configuration. In this case, the left hand side of Eq.~(\ref{energy-rate}) vanishes and we can relate the heat flux from the left reservoir to the left end of the chain, first line of Eq.~(\ref{energy-rate}), with the heat flux from the right end of the chain to the right reservoir, second line of Eq.~(\ref{energy-rate}). Thus, the heat flux from the left reservoir to the right one can be defined as
\begin{align} 
J_E=&\lim_{t \rightarrow \infty}  \left[\frac{1}{2} \left \langle \{\dot X_1(t),F_L(t)\} \right \rangle- \eta_L \left\langle \dot X_1^2(t)\right \rangle \right],  \label{phi1}\\
=&-\lim_{t \rightarrow \infty} \left[\frac{1}{2} \left \langle \{\dot X_N(t),F_R(t)\} \right \rangle- \eta_R \left\langle \dot X_N^2(t)\right \rangle \right].
\end{align}
Because all the normal modes die out for sufficiently long times, the expectation value coincides with the thermal average.

\section{the $Z$-transform \label{Z}}

The scope of this  work is to provide an analytic expression for Eq.~(\ref{phi1}),  hence it is convenient to assume the chain was put in contact to both reservoirs in the far past, i.e., $t\rightarrow-\infty$. Therefore, we can neglect the contribution from the oscillators initial configuration, given that the normal modes decay after some time. This allows us to solve the Eqs.~(\ref{X1}-\ref{XN}) via Fourier instead of Laplace transform in time. Let us denote the Fourier transform of the variables by a tilde on top of them, that is,
\begin{equation}
    \tilde X_j(\omega)=\int\limits_{-\infty}^{\infty}X_j(t)e^{i\omega t} dt\,.
\end{equation}
Therefore, the Heisenberg equations~(\ref{X1}-\ref{XN}) become
\begin{align}
    &\omega^2\tilde{X}_1+\frac{i\omega\eta_L}{m}\tilde{X}_1+\omega_0^2(\tilde X_{2}-\tilde X_1)=-\frac{\tilde F_L(\omega)}{m}, \label{X1-Fourier}
    \\
     &\omega^2\tilde{X}_j+\omega_0^2(\tilde X_{j+1}-2\tilde X_j+\tilde X_{j-1})=0,\quad  j\neq 1, N\,, \label{Xj-Fourier}
   \\ 
    &\omega^2\tilde{X}_N+\frac{i\omega\eta_R}{m}\tilde{X}_N+\omega_0^2(\tilde X_{N-1}-\tilde X_{N})=-\frac{\tilde F_R(\omega)}{m}. \label{XN-Fourier}
\end{align}
Eqs.~(\ref{X1-Fourier}-\ref{XN-Fourier}) correspond to an eigenvalue problem and, in order to solve it, let us define the operator valued analytic function
\begin{equation}
    \mathfrak X(z,\omega):=\sum_{j=1}^{N}\frac{\tilde X_j(\omega)}{z^j}. \label{z-transf}
\end{equation}
This is the finite version of the so-called $Z$-transform, which is used in signal processing. From this definition, $\mathfrak X(z,\omega)$ must vanish at infinity and cannot have any pole of order $N+1$ or higher at $z=0$. The latter imposes that
\begin{equation}
    \varointctrclockwise \mathfrak X(z,\omega)z^{N+n}dz=0, \label{no-pole}
\end{equation}
for any non-negative integer $n$. As a consequence of the Residue Theorem, for any contour $\mathcal C$ that encloses $z=0$, we obtain
\begin{equation}
    \tilde X_j(\omega)=\frac{1}{2\pi i}\varointctrclockwise_{\mathcal C} \mathfrak X(z,\omega)z^{j-1}dz,\quad 1\leq j\leq N. \label{Xj-residue}
\end{equation}
After some algebra, one can show that $\mathfrak X(z,\omega)$ depends only on the endpoint positions, $\tilde X_1$ and $\tilde X_N$ as well as on the fluctuating forces, $\tilde F_L$ and $\tilde F_R$, i.e.,
\begin{align}
   \mathfrak X=\;&\frac{1}{z^2 -(2-\frac{\omega^2}{\omega_0^2})z+1}\left[\left(z-1-\frac{i \eta_L \omega}{m\omega_0^2}\right)\tilde X_1-\frac{\tilde F_L}{m\omega_0^2}\right. \nonumber
   \\
&\left.-\frac{\tilde F_R}{m\omega_0^2\, z^{N-1}}+\left(1-z-\frac{i \eta_R \omega z}{m\omega_0^2}\right)\frac{\tilde X_N}{z^N}\right]. \label{Xz}
\end{align}

At this point, it is convenient to introduce the parametrization $\omega =2\,\omega_0 \sin\frac{\theta}{2}$. This redefinition simplifies the calculation of residues, given that the denominator can be expressed in terms of Chebyshev polynomials of second kind, namely
\begin{equation*}
   \frac{1}{z^2-2z\cos \theta  +1}=\sum_{n=1}^\infty \frac{\sin(n\theta)}{\sin\theta} z^{n-1}\;.
\end{equation*}
From that, we are able to express $\tilde X_j$ in terms of $\tilde X_1$ and $\tilde F_L$, that is,
\begin{align}
 \tilde X_j(\omega)=&\;\frac{\tilde X_1(\omega)}{\cos\frac{\theta}{2}}\left[\cos((j-\tfrac{1}{2})\theta)-\frac{ i \eta_L}{m\omega_0}\sin((j-1)\theta)\right]\nonumber
 \\ 
 &-\frac{\tilde F_L(\omega)}{m\omega_0^2}\, \frac{\sin((N-1)\theta)}{\sin \theta}\;, \quad 2\leq j\leq N. \label{Xj-solution}
\end{align}
 Here, $\theta$ must be viewed as an implicit function of $\omega$. Taking $j=N$ in Eq.~(\ref{Xj-solution}), plugging it into Eq.~(\ref{Xz}), and imposing Eq.~(\ref{no-pole}), we are able to express $\tilde X_1$ solely in terms of $\tilde F_L$ and $\tilde F_R$. Thus, the expression for $\tilde X_1$ can be written as
\begin{equation}
     \tilde X_1(\omega)=\frac{A(\omega) \tilde F_L(\omega)+B(\omega)\tilde F_R(\omega)}{m\omega_0\,\omega D(\omega)}, \label{X1-omega}
\end{equation}
with
\begin{align}
 B(\omega)=&\;-\cos\frac{\theta}{2}\,, \label{B}
 \\
 A(\omega)=&\;2i\alpha_R\sin\Big((N-1)\theta\Big)-\cos\Big(\left(N-\tfrac{1}{2}\right)\theta\Big)\,, \label{A}
 \\
 D(\omega)=&\;\sin(N\theta)+ 2i(\alpha_L+\alpha_R)\cos\Big(\left(N-\tfrac{1}{2}\right)\theta\Big) \nonumber
    \\
    &\;+4\alpha_L\alpha_R\sin\Big((N-1)\theta\Big)\,, \label{D}
\end{align}
where we have defined the dimensionless relaxation constants for each reservoir as $\alpha_{L,R}:=\eta_{L,R}/2m\omega_0$.

\section{Normal modes}

The chain normal modes can be obtained by setting the fluctuating forces to zero in Eqs.~(\ref{Xz}, \ref{Xj-solution}) and imposing the condition (\ref{no-pole}). This implies that the normal mode dispersion is given by $\omega D(\omega)/B(\omega)=0$, which is nothing but the characteristic polynomial coming from Eqs.~(\ref{X1-Fourier}-\ref{XN-Fourier}). This means that $\omega D(\omega)/B(\omega)$ is as a polynomial of order $2N$ in $\omega$. From Eq.~(\ref{Xj-solution}), we can see that the zero-mode $\omega=0$ corresponds to a rigid translation of the whole chain, $\tilde X_j(0)=\tilde X_1(0)$, and do not contribute to the heat transport \footnote{The stationary heat current depends only on $\dot X_1$ and not on $X_1$.}. The other $2N-1$ roots are the zeros of $D(\omega)/B(\omega)$. 

For the sake of simplicity, the analysis presented in this section is restricted only to the case when $\alpha_L=\alpha_R=\alpha$. For the general case, we direct the reader to the supplemental material \cite{SM}. For $\alpha_L=\alpha_R=\alpha$, we have that
\begin{align}
    \frac{D(\omega)}{ B(\omega)}=&\,-\sec\tfrac{\theta}{2}\left[\sin(N\theta)+4\alpha^2\sin\Big((N-1)\theta\Big)\right.\nonumber
    \\
    &+ \left.4i\alpha\cos\Big(\left(N-\tfrac{1}{2}\right)\theta\Big)\right]\,, \label{D-alpha}
\end{align}
which is a quadratic polynomial in $\alpha$. The roots of this polynomial are two complex-valued transcendental equations, whose solutions correspond to the normal mode dispersion without the zero-mode. They can be written explicitly as
\begin{align}
    \alpha&=\frac{i}{2}\left[\sin\tfrac{\theta}{2}+\tan\left(\tfrac{N-1}{2}\theta\right)\cos\tfrac{\theta}{2}\right], \label{alpha-tan}
    \\
    \alpha&=\frac{i}{2}\left[\sin\tfrac{\theta}{2}-\cot\left(\tfrac{N-1}{2}\theta\right)\cos\tfrac{\theta}{2}\right]. \label{alpha-cot}
\end{align}
Note that Eqs.~(\ref{alpha-tan}) and (\ref{alpha-cot}) combined are completely equivalent to $D(\omega)/B(\omega)=0$, when $\alpha_L=\alpha_R=\alpha$. Since Eqs.~(\ref{alpha-tan}, \ref{alpha-cot}) are implicit functions of $\omega$, it is convenient to work directly with $\theta$, keeping in mind that $\omega=2\omega_0\sin\tfrac{\theta}{2}$. 

Because $\alpha$ is real, we can clearly see that Eqs.~(\ref{alpha-tan}, \ref{alpha-cot}) admit no solution for $\theta\in\mathbb R$. This indicates that $D(\omega)/B(\omega)$ has no zeros in the in the real interval $\omega\in[-2\omega_0,2\omega_0]$. In fact, we show in \cite{SM} that, for $\theta=\varphi+i\xi$ and $\varphi\in[-\pi,\pi]$, the zeros of Eq.~(\ref{D-alpha}) must lie inside the region
\begin{equation}
    -\frac{1}{N-1}\ln\left(\left|\csc\tfrac{\varphi}{2}\right|+\sqrt{\csc
^2\tfrac{\varphi}{2}+1}\;\right)<\xi<0. \label{region}
\end{equation}
Moreover, the reality condition $\alpha=\alpha^*$ imposes that
\begin{align}
     \alpha&=\frac{i}{2}\left[\sin\tfrac{\theta}{2}+\tan\left(\tfrac{N-1}{2}\theta\right)\cos\tfrac{\theta}{2}\right]\nonumber
     \\
     &=\frac{i}{2}\left[\sin\left(-\tfrac{\theta^*}{2}\right)+\tan\left(-\tfrac{N-1}{2}\theta^*\right)\cos\left(-\tfrac{\theta^*}{2}\right)\right], \label{alpha-tan2}
    \\
    \alpha&=\frac{i}{2}\left[\sin\tfrac{\theta}{2}-\cot\left(\tfrac{N-1}{2}\theta\right)\cos\tfrac{\theta}{2}\right] \nonumber
    \\
    &=\frac{i}{2}\left[\sin\left(-\tfrac{\theta^*}{2}\right)-\cot\left(-\tfrac{N-1}{2}\theta^*\right)\cos\left(-\tfrac{\theta^*}{2}\right)\right]. \label{alpha-cot2}
\end{align}
Therefore, if $\theta_n$ is a solution of Eq.~(\ref{alpha-tan}), so is $-\theta_n^*$. Equivalently, if $\theta_n$ is a solution of Eq.~(\ref{alpha-cot}), so is $-\theta_n^*$. In terms of $\omega$, this means that if $\omega_n:= \omega(\theta_n)$ is a normal mode frequency, so is $-\omega_n^*$. This is a consequence of $X_j(t)=X^\dagger_j(t)$ and holds true even for the case when $\alpha_L\neq\alpha_R$. To see that, let us express the self-adjoint condition of $X_j(t)$ in terms of its Fourier components $\tilde X_j(\omega)$ and use that $\tilde X_j(\omega)$ is an analytic operator-valued function of $\omega$, that is,
\begin{equation}
     \tilde X_j(\omega)= \tilde X_j^\dagger(-\omega)=\tilde X_j(-\omega^*)
 \end{equation} 
This shows that both $\omega_n$ and $-\omega_n^*$ contribute to the same eigenmode. 
 
Because roots with nonzero real part always come in pairs, $\{\omega_n,-\omega_n^*\}$, there has to be at least one zero on the imaginary axis of the $\omega$ plane. For imaginary values of $\omega$, we have that $\theta=i\xi$. Plugging it into Eqs.~(\ref{alpha-tan}, \ref{alpha-cot}), give us
 \begin{align}
      \alpha&=-\frac{1}{2}\left[\sinh\tfrac{\xi}{2}+\tanh\left(\tfrac{N-1}{2}\xi\right)\cosh\tfrac{\xi}{2}\right], \label{alpha-tanh}
    \\
    \alpha&=-\frac{1}{2}\left[\sinh\tfrac{\xi}{2}+\coth\left(\tfrac{N-1}{2}\xi\right)\cosh\tfrac{\xi}{2}\right]. \label{alpha-coth}
 \end{align}
 
Since the right hand side of Eq.~(\ref{alpha-tanh}) is a monotonically decreasing unbounded function, Eq.~(\ref{alpha-tanh}) always admits solution for $\xi<0$. On the other hand, the function on the right hand side of Eq.~(\ref{alpha-coth}) possesses a global minimum $\alpha_c$ for $\xi<0$ and only admits solution for $\alpha\geq\alpha_c$. The critical value $\alpha_c$ corresponds to the transition between the underdamped case, with only one root on the imaginary axis, and the overdamped case, with 3 zeros on the imaginary axis. It turns out that there is no explicit formula to $\alpha_c$, however from the inequality
\begin{equation}
  1<\frac{\sinh\tfrac{\xi}{2}+\coth\left(\tfrac{N-1}{2}\xi\right)\cosh\tfrac{\xi}{2}}{-\,e^{-\tfrac{1}{2}\xi}}< 1-\frac{2}{(N-1)\xi}\,,
\end{equation}
we can show that
\begin{equation}
    \frac{1}{2}< \alpha_c <\frac{1}{2}\exp\left[\frac{\sqrt{4N-3}-1}{2N-2}\right]\frac{\sqrt{4N-3}+1}{\sqrt{4N-3}-1}\,. \label{alpha-ineq}
\end{equation}

\begin{figure}
\centering
\includegraphics[scale=0.68]{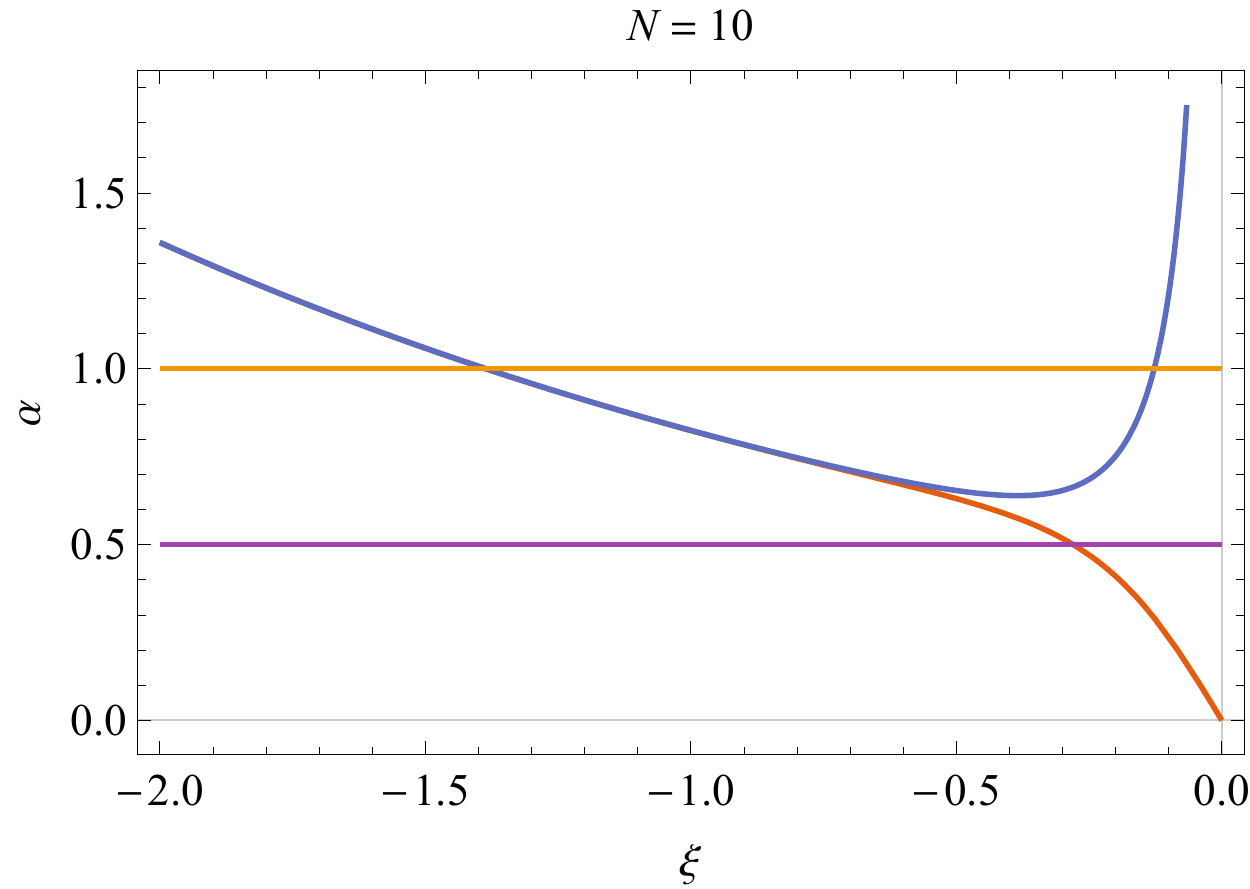}
\caption{Nomal mode frequencies on the imaginary axis.}
\label{Normal mode}
\end{figure}

\section{Large $N$ expansion}

From Eq.~(\ref{region}), we see that the normal mode frequencies which are not on the imaginary axis can be made arbitrarily close to the real axis, the more we increase the number of oscillators. In the limit of large $N$, we can use $1/N$ as the expansion parameter and determine the roots of $D(\omega)/B(\omega)$ perturbatively in powers of $1/N$. This however does not necessarily captures all the roots on the imaginary axis since some of them become non-pertubative, as shown in Fig.~\ref{Normal mode}. Because all non-perturbative zeros must lie on the imaginary axis, we can deal with them separately.
In this $1/N$ expansion, we assume that $\theta$ can be written in the form
\begin{equation}
    \theta=\sum_{k=1}^\infty\frac{\vartheta_k}{N^k}. \label{theta-N}
\end{equation}
Plugging Eq.~(\ref{theta-N}) into the expression $D(\omega)/B(\omega)$ and expanding it in powers of $1/N$ give us
\begin{align}
\frac{D(\omega)}{B(\omega)}&=\cos\vartheta_1\Big\{\Big[\left(1+4\alpha_L\alpha_R\right)\tan\vartheta_1+2i(\alpha_L+\alpha_R)\Big]\nonumber
\\
&+\frac{1}{N}\left[\vartheta_2\,\Big(1+4\alpha_L\alpha_R-2i(\alpha_L+\alpha_R)\tan\vartheta_1\Big)\right.\nonumber
\\
&+\left.\left.\vartheta_1\Big(i(\alpha_L+\alpha_R)\tan\vartheta_1-4\alpha_L\alpha_R\Big)\right]\right\}+\mathcal O\left(N^{-2}\right).
\end{align}

The normal mode dispersion is obtained by equating each coefficient of this expansion to  zero. Here, we restrict the analysis only to  the first two terms. The leading order coefficient gives us
\begin{figure*}
\centering
\includegraphics[width=\textwidth]{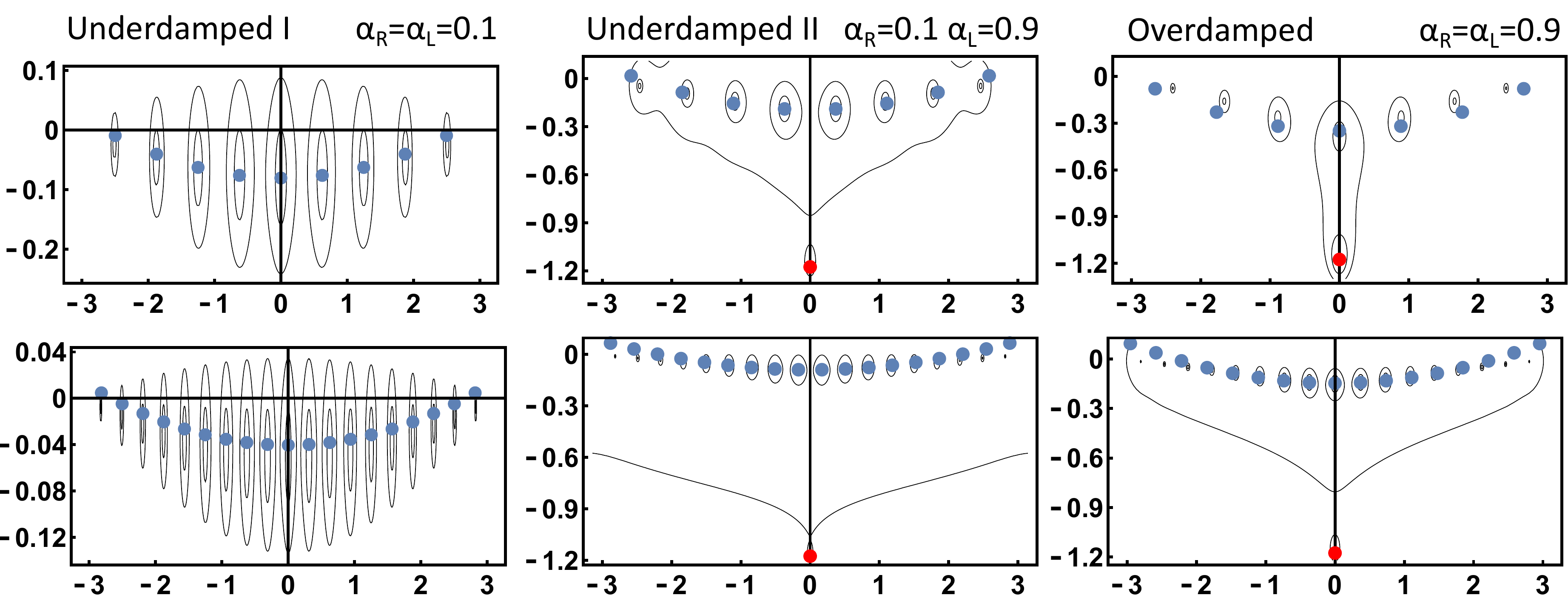}
\caption{Contour plot of equation (\ref{D-alpha}) with the roots approximated by the perturbative technique developed in section V up to $\mathcal{O}(N^{-3})$. The first line is plotted for $N=5$ and the second for $N=10$ and the blue (red) points are the perturbative (non-perturbative) roots. The approximation up to $\mathcal{O}(N^{-2})$ of $\theta_n$ is linear and not able to account for the inflected curve of the distribution of poles.}
\label{Polesfig}
\end{figure*}
\begin{equation}
    \tan\vartheta_1=-\frac{2i(\alpha_L+\alpha_R)}{1+4\alpha_L\alpha_R}=i\left[\frac{(\alpha_L-\frac{1}{2})(\alpha_R-\frac{1}{2})}{\alpha_L\alpha_R+\frac{1}{4}}-1\right], \label{damping}
\end{equation}
and at order $1/N$, we find
\begin{align}
    \vartheta_2&=\frac{4\alpha_L\alpha_R-i(\alpha_L+\alpha_R)\tan\vartheta_1}{1+4\alpha_L\alpha_R-2i(\alpha_L+\alpha_R)\tan\vartheta_1}\,\vartheta_1\,,\nonumber
    \\
     \vartheta_2&=\frac{16\alpha^2_L\alpha^2_R-2(\alpha^2_L+\alpha^2_R)}{(4\alpha_L^2-1)(4\alpha_R^2-1)}\,\vartheta_1\,.
\end{align}
In the last line, we have replaced $\tan\vartheta_1$ by the right hand side of Eq.~(\ref{damping}).

The expression~(\ref{damping}) only admits solution on the imaginary axis when $\alpha_L,\alpha_R<1/2$ or $\alpha_L,\alpha_R>1/2$. In fact, when $\max(\alpha_L,\alpha_R)>1/2$ and $\min(\alpha_L,\alpha_R)<1/2$, the single root on the imaginary axis become non-perturbative and cannot be captured by the ansatz~(\ref{theta-N}). Moreover, for $\alpha_L,\alpha_R>1/2$, there are 3 zeros on the imaginary axis and 2 of them become nonpertubative. In the following, we consider the three cases separately.

\subsection{Underdamped case 1: $\alpha_L<\frac{1}{2}$ and $\alpha_R<\frac{1}{2}$}

In this case, all roots are perturbative and we must get $2N-1$ solutions from Eq.~(\ref{damping}). Thus, solving for $\vartheta_1$,
\begin{equation}
    \vartheta_1=n\pi-i\delta, 
\end{equation}
for $n\in\mathbb Z$, such that $-N+1\leq n\leq N-1$ and 
\begin{equation}
    \delta:=\text{arctanh}\left[\frac{2(\alpha_L+\alpha_R)}{1+4\alpha_L\alpha_R}\right].
\end{equation}

Therefore, the roots $\theta_n$ are of the form
\begin{equation}
    \theta_n=\left[1+\frac{16\alpha^2_L\alpha^2_R-2(\alpha^2_L+\alpha^2_R)}{N(4\alpha_L^2-1)(4\alpha_R^2-1)}\right]\left(\frac{n\pi}{N}-i\frac{\delta}{N}\right).
\end{equation}

\subsection{Underdamped case 2: $\min(\alpha_L,\alpha_R)<\frac{1}{2}$ and $\max(\alpha_L,\alpha_R)>\frac{1}{2}$}

In this case, the imaginary root is nonperturbative and can be written as
\begin{equation}
 \theta_{+}=-2i\ln(2\max(\alpha_L,\alpha_R)).   
\end{equation}
For details check the SM \cite{SM}. The other $2N-2$ normal mode frequencies are the solutions of Eq.~(\ref{damping}), namely
\begin{equation}
    \vartheta_1=\left(n+\tfrac{1}{2}\right)\pi-i\tilde\delta, 
\end{equation}
for $n\in\mathbb Z$, such that $-N+1\leq n\leq N-2$ and 
\begin{equation}
    \tilde\delta:=\text{arccoth}\left[\frac{2(\alpha_L+\alpha_R)}{1+4\alpha_L\alpha_R}\right].
\end{equation}
The zeros for this case are
\begin{equation}
    \theta_n=\left[1+\frac{16\alpha^2_L\alpha^2_R-2(\alpha^2_L+\alpha^2_R)}{N(4\alpha_L^2-1)(4\alpha_R^2-1)}\right]\left(\frac{n+\frac{1}{2}}{N}\,\pi-i\frac{\tilde\delta}{N}\right).
\end{equation}

\subsection{Overdamped case: $\alpha_L>\frac{1}{2}$ and $\alpha_R>\frac{1}{2}$}

Here, two of the imaginary roots are non-perturbative and given by
\begin{equation}
    \theta_L=-2i\ln(2\alpha_L)\quad \text{and}\quad \theta_R=-2i\ln(2\alpha_R).
\end{equation}
Again, for details we refer to \cite{SM}. The other $2N-3$ roots are the same as in the case when $\alpha_L,\alpha_R<1/2$, with the exception of the pair of endpoint zeros, that is,
\begin{equation}
    \theta_n=\left[1+\frac{16\alpha^2_L\alpha^2_R-2(\alpha^2_L+\alpha^2_R)}{N(4\alpha_L^2-1)(4\alpha_R^2-1)}\right]\left(\frac{n\pi}{N}-i\frac{\delta}{N}\right),
\end{equation}
for $n\in\mathbb Z$, such that $-N+2\leq n\leq N-2$.

\section{Thermal conductance}

In this section, we provide an exact expression for the heat conductance of a harmonic chain. Again, for calculation details we refer to \cite{SM}. Plugging expression (\ref{X1-omega}) into Eq.~(\ref{phi1}) and using Eq.~(\ref{Noise}), we end up with  with the following expression for the heat current  
\begin{align}
    J_E=\frac{\hbar\alpha_L\alpha_R}{\pi}\int\limits_{-\infty}^{\infty} &d\omega\; \frac{\omega B(\omega)B(-\omega)}{D(\omega)D(-\omega)} \left[ \coth\left(\frac{\hbar\omega}{2k_BT_R}\right)   \right.\nonumber
    \\
    &-\left.\coth\left(\frac{\hbar\omega}{2k_BT_L}\right) \right].\label{heat-flux}
\end{align}
Eq.~(\ref{heat-flux}) presented this way is physically intuitive, since the term in square brackets is nothing but the difference between the Bose-Einstein distribution function of the reservoirs and the term multiplying it accounts for the normal mode contribution, that is, 
\[
\frac{B(\omega)}{D(\omega)}\times\frac{B(-\omega)}{D(-\omega)}\propto\prod_{n=1}^{2N-1}\frac{1}{\omega^2-\omega_n^2},
\]
where $\omega_n$ are the normal mode frequencies with the zero-mode removed.
\begin{figure}
\centering
\includegraphics[scale=0.45]{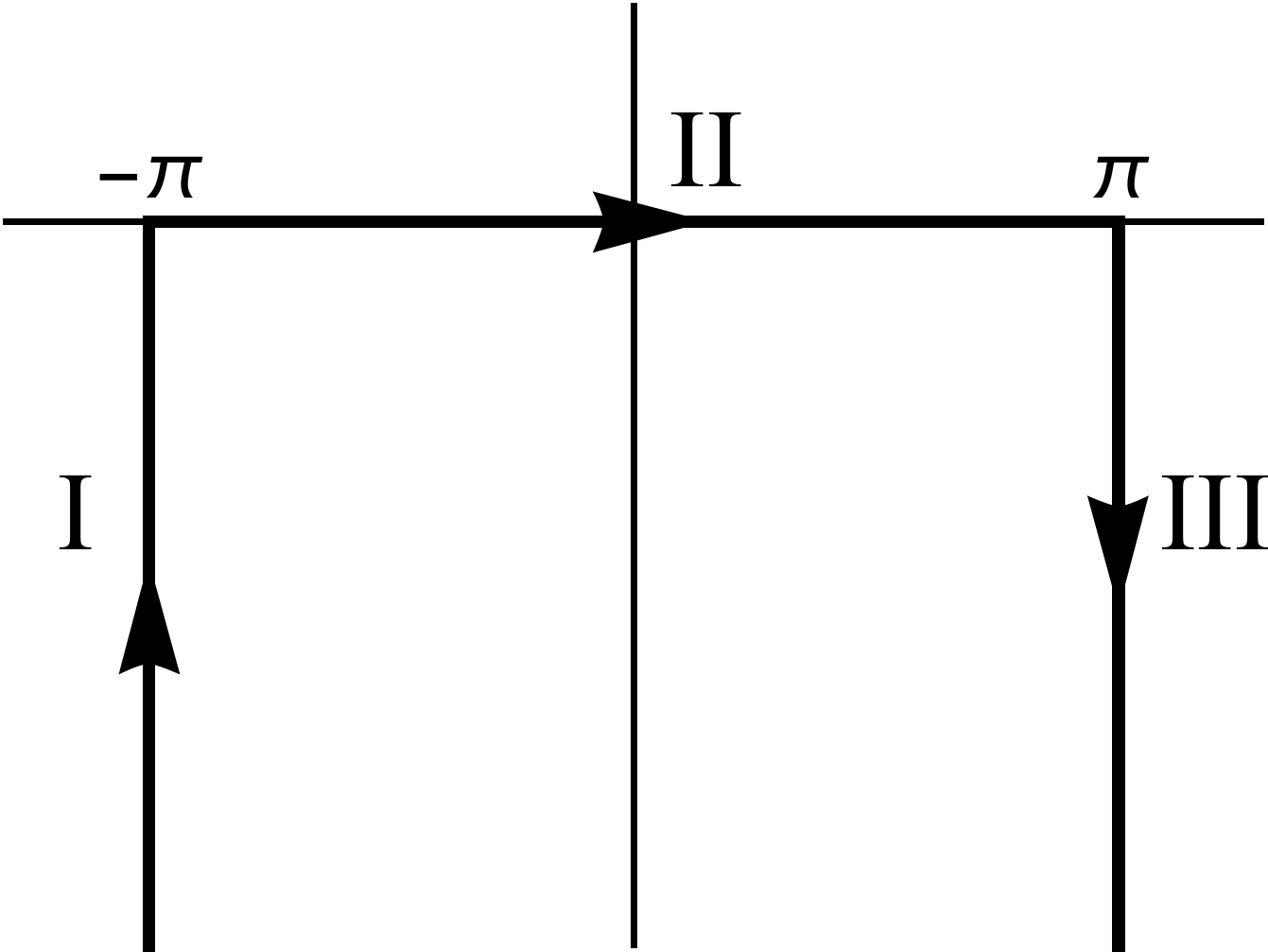}
\caption{Countour integration $\mathcal C$ in the complex $\theta$-plane.}
\label{countour}
\end{figure}

The heat current (\ref{heat-flux}) is an integral over the whole real line in the variable $\omega$. However, the integrand is an implicit function of $\omega$ and it becomes convenient to express the integral (\ref{heat-flux}) in terms of $\theta$. This means that the integration contour is a curve in the $\theta$ complex plane, given that there is no real value of $\theta$ that satisfies $|\omega|> 2 \omega_0$. Expressing $\theta$ as $\varphi+i\xi$, we have that
\begin{equation}
    \omega=2\omega_0\left[\sin\left(\tfrac{\varphi}{2}\right)\cosh\left(\tfrac{\xi}{2}\right)+i\cos\left(\tfrac{\varphi}{2}\right)\sinh\left(\tfrac{\xi}{2}\right)\right]. \label{omega-theta}
\end{equation}
From this, we can see that both $\theta=-\pi\pm i\xi$, with $\xi>0$, map into the negative branch cut $(-\infty,-2\omega_0)$. Equivalently, the positive branch cut $(2\omega_0, +\infty)$ also has two preimages, given by $\theta=\pi\pm i\xi$, with $\xi>0$.  Therefore, there are 4 possible ways to define the integral contour in the $\theta$ complex plane and the final result must be independent of this choice. In this work, we choose the contour $\mathcal C$,  shown in Fig.~\ref{countour}.

In terms of $\theta$, the expression for the heat flux becomes
\begin{widetext}
\begin{align} \label{JE}
J_E=&\,\frac{4\hbar\omega_0^2\alpha_L\alpha_R}{\pi}\int_{\mathcal C}d\theta\,\frac{\sin\tfrac{\theta}{2}\cos^3\tfrac{\theta}{2}\left[\coth\left(\frac{\hbar\omega_0}{k_BT_R}\sin(\tfrac{\theta}{2})\right)-\coth\left(\frac{\hbar\omega_0}{k_BT_L}\sin(\tfrac{\theta}{2})\right)\right]}{\left[ \sin( N \theta) + 4\alpha_L \alpha_R \sin((N-1)\theta) \right]^2+4(\alpha_L+\alpha_R)^2\cos^2\left(\left(N-\tfrac{1}{2}\right) \theta\right)}\,.
\end{align}
\end{widetext}

It is not hard to show that Eq. (\ref{JE}) can be expressed either as the sum over all the integrand residues inside the contour $\mathcal C$, or broken up into two real integrals. Here it is convenient to introduce the Debye temperature $T_D:=2\hbar\omega_0/k_B$. This definition differs slightly from the one in textbooks, since we are neglecting the contribution from transverse phonons. 

Notice that the full expression for the heat flux (\ref{JE}) is not proportional to the temperature difference $\Delta T=T_R-T_L$, but it is instead proportional to the difference between the Bose-Einstein distributions of each reservoir. Nevertheless, when the temperature difference $\Delta T$ is much smaller then the average temperature $T=\frac{1}{2}(T_L+T_R)$ and the Debye temperature $T_D$, the heat current~(\ref{JE}) becomes approximately proportional to $\Delta T$, what allows us to introduce the one-dimensional {\it thermal conductance} as
\begin{widetext}
\begin{align} \label{conductance}
\mathcal K=\,\frac{k_B^2T_D}{2\pi\hbar}\frac{\alpha_L\alpha_R }{ \Theta^2}&\int_{\mathcal C}d\theta\,\frac{\sin^2\tfrac{\theta}{2} \cos^3\tfrac{\theta}{2}\,\text{csch}^2\Big[\frac{1}{2\Theta}\sin(\tfrac{\theta}{2})\Big]}{\left[\sin( N \theta) + 4\alpha_L \alpha_R\sin((N-1)\theta) \right]^2+4(\alpha_L+\alpha_R)^2\cos^2\left(\left(N-\tfrac{1}{2}\right) \theta\right)}\,,
\end{align}
\end{widetext}
Here, we have introduced the reduced temperature scale $\Theta=T/T_D$, which describes the system in the quantum regime when $\Theta\ll 1$ or classical limit $\Theta\sim 1$.
In order to study the heat conductance without specific considerations about the materials, it is useful to define dimensionless heat conductance $\widetilde{\mathcal{K}}=\pi\hbar\mathcal{K}/k_B^2T_D$ and express (\ref{conductance}) as the sum of two real integrals as
\begin{widetext}  
\begin{align}
\widetilde{\mathcal{K}}=\frac{\alpha_L \alpha_R}{\Theta^2}&\left[\int\limits_{0}^\pi d\varphi\,\frac{\sin^2\tfrac{\varphi}{2}\cos^3\tfrac{\varphi}{2}\,\text{csch}^2\Big[\frac{1}{2\Theta}\sin(\tfrac{\varphi}{2})\Big]}{\left[\sin( N \varphi) + 4\alpha_L \alpha_R\sin((N-1)\varphi) \right]^2+4(\alpha_L+\alpha_R)^2\cos^2\left(\left(N-\tfrac{1}{2}\right) \varphi\right)}\nonumber \right.
   \\
    &\left.-\int\limits_{0}^\infty d\xi\,\frac{\sinh^3\tfrac{\xi}{2}\cosh^2\tfrac{\xi}{2}\,\text{csch}^2\Big[\frac{1}{2\Theta}\cosh(\tfrac{\xi}{2})\Big]}{\left[\sinh (N \xi) - 4\alpha_L \alpha_R\sinh((N-1)\xi) \right]^2-4(\alpha_L+\alpha_R)^2\sinh^2\left(\left(N-\tfrac{1}{2}\right) \xi\right)}\right] \,.\label{conductance3}
\end{align}
\end{widetext}

As one can see from Fig. \ref{HeatCond}, the thermal conductance of this model saturates for large $N$, which refers to a ballistic transport. This was somewhat expected, since for large $N$ the normal mode frequencies approach the real axis, which a characteristic of ballistic transport. {\it Hence, we obtained through the analytic expression (\ref{conductance}) that the Fourier law cannot be achieved for a harmonic chain without disorder}. This corroborates a well-known result in the literature of $1$D systems (see, for example, \cite{Lebowitz}, and references therein) which states that, due to the existence of several conserved quantities, integrable systems always present ballistic heat transport. The so-called normal transport, $J_{E}\sim N^{-1}$, results from the non-conservation of linear momentum due either to a one-body pinning potential or inelastic (dissipative) effects. Here we should remark that the dissipative motion of the endpoints of our harmonic chain affects at most its boundary effects.
\begin{figure}[h!]
\centering
\includegraphics[scale=0.8]{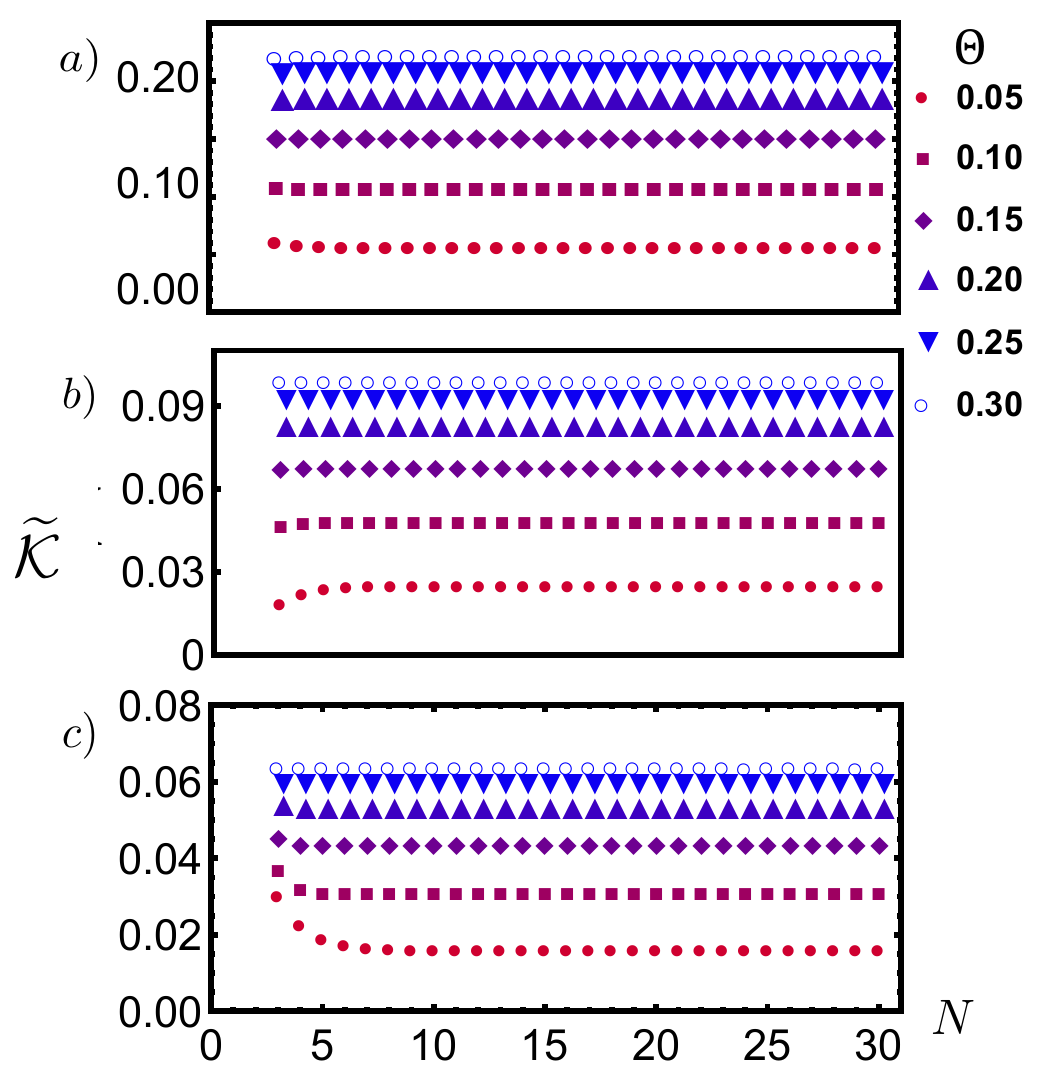}
\caption{Heat conductance given by equation (\ref{conductance3}) as a function of $N$ for different temperatures and damping parameters, i.e., a) $\alpha_L=\alpha_R=0.2$ , b) $\alpha_L=0.2$ and $\alpha_R=5$ and c) $\alpha_L=\alpha_R=5$.}
\label{HeatCond}
\end{figure}

\section{conclusions}

In this contribution we have analyzed the heat transport between two thermal reservoirs kept at different temperatures via a one-dimensional harmonic chain whose endpoints are immersed in each reservoir.   We have  modeled the latter by Brownian particles with different damping constants. 

Making use of the Z-transform we were able to  provide an exact solution to the problem which is valid for any temperature $T$ of each  reservoir, and number  $N$ of particles in the chain.  This solution is expressed as an ordinary  integral  which depends on the  thermal occupation number  of the harmonic modes coupled to the reservoirs and their frequency distribution.

Numerical solutions of this integral for  a small temperature drop between the two reservoirs have shown us that the energy transport is ballistic for large $N$ at any typical temperature scale  of the reservoirs. In other words, the thermal conductance turns out to be independent of $N$. Actually, this signature persists even for small $N$ and high temperatures ($T \gtrsim T_D$).  The thermal conductance only presents a  very modest dependence on  $N$ in the extreme quantum regime ($T<<T_D$) for $N \lesssim 10$, and, particularly, for the overdamped case when the motion of the endpoints of the chain is heavily damped. 

In the absence of the thermal reservoirs, energy is transferred from one end of the chain to the other  through the channels provided by its normal modes, namely, its longitudinal phonons. For example, an energy pulse is transmitted through a coherent superposition of  phonons with different  frequencies.  When the chain is coupled to the reservoirs, these normal modes acquire imaginary parts which are reminiscent of the existence of the damped particles at the endpoints of the chain. No matter what the real part of these frequencies are (they range roughly from  $2\omega_0$ to $\pi\omega_0/N$), their imaginary parts are  $\mathcal{O}(1/N)$.  Consequently, for  macroscopic chains (very large $N$), phonons are undamped and the energy transport is ballistic. The dependence of the thermal conductance on temperature and damping parameters can be easily understood as we analyze the transfer of energy from one reservoir to the other.

Suppose a few quanta of energy leave the left reservoir which is supposed to be at higher temperature than the right reservoir. This is firstly accomplished  by depopulating the equilibrium distribution of the particle at the left endpoint of the chain. Subsequently, by the reasoning we have just employed above, they are ballistically transported to the right reservoir and tend to overpopulate the equilibrium distribution of the particle at the right endpoint of the chain. Since both endpoints are connected to thermal reservoirs, this depopulation (overpopulation) of the left (right) reservoir is compensated by the absortion (emission) of the missing (excess) energy from (to) the left (right) reservoir. These effects involving the equilibrium distribution of energy of the particles at the endpoints of the chain are the only ones responsible for the temperature and damping dependence of the thermal conductance.

The argument we have used above is particularly valid for $N\gg1$ because, in this case, we can assume that the equilibrium distribution of the particles at the endpoints of the chain are very weakly affected by its remaining components. In other words, the dependence  we have discussed is mainly a boundary effect. However, this ceases to be true as we reduce the number of particles of the chain because now the few eigenfrequencies of the system can be quite different from the natural frequencies of oscillation of the endpoint particles. This explains that modest dependence of the thermal conductance on $N$ which shows up only for $N\lesssim 10$ and low temperatures.

Once this is understood we can safely claim that this model, although very useful and interesting to be studied, cannot account for the well-established Fourier law for thermal transport in macroscopic systems. In order to achieve it, we need to modify the present model by introducing either disorder or inelastic effects to destroy the ballistic character of the transport along the chain. In analogy with electric transport through metallic  wires, the former (latter) choice would be equivalent to the presence of elastic (inelastic) resistance along the wire.

 Finally, it would be desirable to compare our results with those measured in experiments of heat transport in nanoscopic systems. A very good candidate to be studied would be, for example, gold single-chain nanowires \cite{nano}, which present a geometrical structure very similar to our theoretical model. Nevertheless, it is quite well-known that thermal transport in these systems are mainly due to electrons, contrary to our present analysis of phonon transport. Even if one is careful enough to shield the system from the presence of electric fields, they will be naturally generated in metallic systems through thermoelectric effects. Therefore, for a faithful comparison with our findings, we need to either find an appropriate device of  insulating material or rely on the ingenuity of experimentalists  to properly select the effects of phonon transport in tiny metallic chains. Moreover, to the best of our knowledge, these experiments are mostly performed with the set up lying on a substrate, which cannot be represented by our model. Instead, these are more suitable for being modelled by a chain of damped oscillators, a natural extension of the present work that we plan to pursue in the near future.

{\bf Acknowledgements.} Financial support for this work was provided by the Brazilian funding agencies, CAPES and CNPq, under the projects numbers 140553/2018-5 and 302420/2015-0, respectively. GMM
thanks Funda\c c\~ao de Amparo \`a Pesquisa do Estado de
S\~ao Paulo (FAPESP) for financial support under grant 2016/13517-0.



\onecolumngrid
\clearpage

\section*{Supplementary Material for: ``Heat Transport in Mesoscopic Chains''}

\setcounter{section}{0}

 \section{Normal Modes}
 
 As stated in the main text, the normal mode dispersion is obtained by the zeros of $\omega D(\omega)/B(\omega)$. The zero-mode $\omega=0$ corresponds to a rigid translation of the chain and do not contribute to the heat transport. Here, we will focus on the $2N-1$ zeros of 
 \begin{equation}
    \frac{D(\omega(\theta))}{B(\omega(\theta))}=-\sec\left(\tfrac{\theta}{2}\right)\left[\sin(N\theta)+4\alpha_L\alpha_R\sin((N-1)\theta)+2i(\alpha_L+\alpha_R)\cos\left(\left(N-\tfrac{1}{2}\right)\theta\right)\right]\,.  \label{modes}
 \end{equation}
Since Eq.~(\ref{modes}) is an explicit function of $\theta$, it is convenient to work directly in the $\theta$-complex plane. Because $\sin\theta$ is a periodic function, its inverse is a multi-valued function and we have to specify the $\theta$ domain in which the function is one-to-one. From Eq.~(\ref{omega-theta}) of the main text, that is,
\begin{equation*}
    \omega(\theta)=2\omega_0\left[\sin\left(\tfrac{\varphi}{2}\right)\cosh\left(\tfrac{\xi}{2}\right)+i\cos\left(\tfrac{\varphi}{2}\right)\sinh\left(\tfrac{\xi}{2}\right)\right],
\end{equation*}
we recover almost the whole $\omega$ complex plane, with the exception of the branch cuts $(-\infty,-2\omega_0]$ and $[2\omega_0, +\infty)$, for  $\varphi\in(-\pi,\pi)$ and $\xi\in \mathbb R$. 
We can see that both $\theta=-\pi\pm i\xi$, with $\xi\geq0$, map into the negative branch cut $(-\infty,-2\omega_0]$ and, equivalently, the positive branch cut $[2\omega_0, +\infty)$ also have two preimages, given by $\theta=\pi\pm i\xi$, with $\xi\geq0$. Therefore, we can restrict our analysis to the complex strip $\varphi\in[-\pi,\pi]$ and $\xi\in \mathbb R$, keeping in mind that each branch cut has two preimages in this domain.

To show that all the normal modes are damped (with the exception of the zero-mode), we need to prove that all zeros of $D(\omega)/B(\omega)$ lie on the $\omega$ lower half-plane or, in other words, on the domain $-\pi<\varphi<\pi$ and $\xi<0$.

\begin{thm}
All the zeros of $D(\omega(\theta))/B(\omega(\theta))$, defined in Eq.~(\ref{modes}) lie on the lower semi-strip:  $(-\pi,\pi)\times\mathbb R_-^*$. \label{thm1}
\end{thm}

\begin{proof} Let us first show that the expression $(\ref{modes})$ has no zeros for $-\pi\leq\varphi\leq\pi$ and $\xi=0$. Plugging $\theta=\varphi$ into Eq.~(\ref{modes}), we obtain that
\begin{equation}
    \frac{D(\omega(\varphi))}{B(\omega(\varphi))}=-\left[\sin((N-\tfrac{1}{2})\varphi)(1+4\alpha_L\alpha_R)+\frac{\cos\left(\left(N-\tfrac{1}{2}\right)\varphi\right)}{\cos\tfrac{\varphi}{2}}(1-4\alpha_L\alpha_R)\right]+i\left[2\,(\alpha_L+\alpha_R)\frac{\cos\left(\left(N-\tfrac{1}{2}\right)\varphi\right)}{\cos\tfrac{\varphi}{2}}\right]\neq0\,. \label{real-line}
\end{equation}
Although both $D(\omega)$ and $B(\omega)$ vanish for $\varphi=\pm \pi$, the function $D(\omega)/B(\omega)$ is regular for $\varphi=\pm \pi$, since
\[
\cos\left((2N-1)\frac{\varphi}{2}\right)=\sum_{n=0}^{N-1}\begin{pmatrix}
  2N-1 \\
 2n \\
\end{pmatrix}\cos^{2(N-n)-1}\left(\tfrac{\varphi}{2}\right)\sin^{2n}\left(\tfrac{\varphi}{2}\right)\,.
\]

It is also straightforward to show that Eq.~(\ref{modes}) has no roots for $\theta=\pm\pi+i\xi$, with $\xi\neq0$, since
\begin{equation}
     \frac{D\big(\omega(\pm\pi+i\xi)\big)}{B\big(\omega(\pm\pi+i\xi)\big)}=\pm i\frac{(-1)^{N+1}}{\sinh\left(\tfrac{\xi}{2}\right)}\bigg[\mp\,2\,(\alpha_L+\alpha_R)\sinh\left(\left(N-\tfrac{1}{2}\right)\xi\right)+i\,\Big(\sinh(N\xi)-4\alpha_L\alpha_R\sinh((N-1)\xi)\Big)\bigg]\neq0\,. \label{branch-cuts}
\end{equation}
As a matter of fact, Eqs.~(\ref{real-line}, \ref{branch-cuts}) demonstrate that there are no zeros of Eq.~(\ref{modes}) for real values of $\omega$. 

Let us now show that there are no zeros of Eq.~(\ref{modes}) for $\varphi\in(-\pi,\pi)$ and $\xi>0$. Since $\cos\tfrac{\theta}{2}$ has no poles or zeros in this domain, all the roots of $D(\omega)/B(\omega)$ will correspond to zeros of $D(\omega)$. Moreover, because $\cos\left(\left(N-\tfrac{1}{2}\right)\theta^*\right)$ does not vanish for $\xi\neq0$, the zeros of $\cos\left(\left(N-\tfrac{1}{2}\right)\theta
^*\right)D(\omega)$ for $\xi\neq0$ correspond to the roots of $D(\omega)$. Using that
\begin{align}
    \sin\left(\left(N-\tfrac{1}{2}\right)\theta^*\right)D(\omega)=&\;\frac{1}{2}\left[\sin\tfrac{\varphi}{2}\cosh\left(\left(2N-\tfrac{1}{2}\right)\xi\right)+\cosh\tfrac{\xi}{2}\sin\left(\left(2N-\tfrac{1}{2}\right)\varphi\right)+4\alpha_L\alpha_R\left(\cosh\tfrac{\xi}{2}\sin\left(\left(2N-\tfrac{3}{2}\right)\varphi\right)\right.\right.\nonumber
    \\
    &\left.-\,\sin\tfrac{\varphi}{2}\cosh\left(\left(2N-\tfrac{3}{2}\right)\xi\right)\Big)\right]+\frac{i}{2}\left[\cos\tfrac{\varphi}{2}\sinh\left(\left(2N-\tfrac{1}{2}\right)\xi\right)+\sinh\tfrac{\xi}{2}\cos\left(\left(2N-\tfrac{1}{2}\right)\varphi\right)\right.\nonumber
    \\
    &+4\alpha_L\alpha_R\left(\cos\tfrac{\varphi}{2}\sinh\left(\left(2N-\tfrac{3}{2}\right)\xi\right)-\sinh\tfrac{\xi}{2}\cos\left(\left(2N-\tfrac{3}{2}\right)\varphi\right)\right)+2\,(\alpha_L+\alpha_R)\nonumber
    \\
    &\times\Big(\cosh\left(\left(2N-1\right)\xi\right)+\cos\left(\left(2N-1\right)\varphi\right)\Big)\bigg]\,.
\end{align}
and considering each term of the the imaginary part separately, we obtain that
\begin{align}
    &\cos\tfrac{\varphi}{2}\sinh\left(\left(2N-\tfrac{1}{2}\right)\xi\right)+\sinh\tfrac{\xi}{2}\cos\left(\left(2N-\tfrac{1}{2}\right)\varphi\right)>\sinh\tfrac{\xi}{2}\left[\left(4N-1\right)\cos\tfrac{\varphi}{2}+\cos\left(\left(2N-\tfrac{1}{2}\right)\varphi\right)\right]>0\,,
\\
   &\cos\tfrac{\varphi}{2}\sinh\left(\left(2N-\tfrac{3}{2}\right)\xi\right)-\sinh\tfrac{\xi}{2}\cos\left(\left(2N-\tfrac{3}{2}\right)\varphi\right)>\sinh\tfrac{\xi}{2}\left[\left(4N-3\right)\cos\tfrac{\varphi}{2}-\cos\left(\left(2N-\tfrac{3}{2}\right)\varphi\right)\right]>0\,,
\\
    &2\,(\alpha_L+\alpha_R)\left(\cosh\left(\left(2N-1\right)\xi\right)+\cos\left(\left(2N-1\right)\varphi\right)\right)>2\,(\alpha_L+\alpha_R)\left(1+\cos\left(\left(2N-1\right)\varphi\right)\right)>0\,,
\end{align}
for $\xi>0$. This implies that 
\begin{equation}
    \Im\left[\sin\left(\left(N-\tfrac{1}{2}\right)\theta^*\right)D(\omega)\right]>0\qquad \text{for}\qquad \xi>0\,, 
\end{equation}
and ends our proof.
\end{proof}

From the self-adjoint condition and analyticity of the operators $X_j(\omega)$, we argued in the main text that the phonon frequencies which do not lie on the imaginary axis must come in pairs. To see that, let us expand $D(\omega)$ in its real and imaginary parts, that is,
 \begin{align}
     D(\omega)=&\,\bigg[\sin(N\varphi)\cosh(N\xi)+4\alpha_L\alpha_R\sin\left((N-1)\varphi\right)\cosh\left((N-1)\xi\right)+2\,(\alpha_L+\alpha_R)\sin\left((N-\tfrac{1}{2})\varphi\right)\sinh\left((N-\tfrac{1}{2})\xi\right)\bigg] \nonumber
     \\
     +\,i&\bigg[\cos(N\varphi)\sinh(N\xi)+4\alpha_L\alpha_R\cos\left((N-1)\varphi\right)\sinh\left((N-1)\xi\right)+2\,(\alpha_L+\alpha_R)\cos\left((N-\tfrac{1}{2})\varphi\right)\cosh\left((N-\tfrac{1}{2})\xi\right)\bigg].\label{modes2}
 \end{align}
It is straightforward to see that, if $\theta_n=\varphi_n+i\xi_n$ is a root of Eq.~(\ref{modes2}), so is $-\theta_n^*=-\varphi_n+i\xi_n$. Because these roots come in pairs and $D(\omega)$ has $2N-1$ zeros for $\xi<0$, there must exist at least one root on the negative imaginary axis. The normal mode frequencies with $\varphi=0$ are solutions of
\begin{equation}
    \sinh(N\xi)+4\alpha_L\alpha_R\sinh\left((N-1)\xi\right)+2\,(\alpha_L+\alpha_R)\cosh\left((N-\tfrac{1}{2})\xi\right)=0\,. \label{roots-Im}
\end{equation}

Here, it is convenient to define the following variables
\begin{equation}
    \min(\alpha_L,\alpha_R):=\alpha\,,
    \qquad\qquad\qquad \max(\alpha_L,\alpha_R):=\nu\alpha\,,
\end{equation}
where $\nu\geq1$. Plugging these into Eq.~(\ref{roots-Im}), we obtain
\begin{equation}
     \sinh(N\xi)+4\nu\alpha^2\sinh\left((N-1)\xi\right)+2\alpha(1+\nu)\cosh\left((N-\tfrac{1}{2})\xi\right)=0\,, \label{roots-Im2}
\end{equation}
which can be viewed as a quadratic polynomial in $\alpha$ and each solution of this quadratic equation give us one transcendental equation in $\xi$, that is,
 \begin{align}
     \alpha&=\frac{1}{2\nu\sinh\left((N-1)\xi)\right)}\left[-\frac{1+\nu}{2}\cosh\left(\left(N-\tfrac{1}{2}\right)\xi\right)+\sqrt{\nu\cosh^2\tfrac{\xi}{2}+\tfrac{1}{4}(\nu-1)^2\cosh^2\left(\left(N-\tfrac{1}{2}\right)\xi\right)}\right], \label{alpha1}
     \\
     \alpha&=\frac{1}{2\nu\sinh\left((N-1)\xi)\right)}\left[-\frac{1+\nu}{2}\cosh\left(\left(N-\tfrac{1}{2}\right)\xi\right)-\sqrt{\nu\cosh^2\tfrac{\xi}{2}+\tfrac{1}{4}(\nu-1)^2\cosh^2\left(\left(N-\tfrac{1}{2}\right)\xi\right)}\right]. \label{alpha2}
 \end{align}
 
 Because the right hand side of Eq.~(\ref{alpha1}) is a monotonically decreasing function, this transcendental equation always admits one solution. On the other hand, the right hand side of Eq.~(\ref{alpha2}) has a global minimum for $\xi<0$, the smallest value of $\alpha$ so that any solution exists. To indicate that, let us analyze the behavior of both functions for $\xi\rightarrow-\infty$ and $\xi\rightarrow0^-$, that is,
 \begin{align}
     &\lim_{\xi\rightarrow-\infty}\left[-\frac{1+\nu}{4\nu}\frac{\cosh\left(\left(N-\tfrac{1}{2}\right)\xi\right)}{\sinh\left((N-1)\xi)\right)}+\frac{\sqrt{\nu\cosh^2\tfrac{\xi}{2}+\tfrac{1}{4}(\nu-1)^2\cosh^2\left(\left(N-\tfrac{1}{2}\right)\xi\right)}}{2\nu\sinh\left((N-1)\xi)\right)}\right]=\frac{e^{-\frac{\xi}{2}}}{2\nu},
     \\
      &\lim_{\xi\rightarrow0^-}\left[-\frac{1+\nu}{4\nu}\frac{\cosh\left(\left(N-\tfrac{1}{2}\right)\xi\right)}{\sinh\left((N-1)\xi)\right)}+\frac{\sqrt{\nu\cosh^2\tfrac{\xi}{2}+\tfrac{1}{4}(\nu-1)^2\cosh^2\left(\left(N-\tfrac{1}{2}\right)\xi\right)}}{2\nu\sinh\left((N-1)\xi)\right)}\right]=0,
     \\
      &\lim_{\xi\rightarrow-\infty}\left[-\frac{1+\nu}{4\nu}\frac{\cosh\left(\left(N-\tfrac{1}{2}\right)\xi\right)}{\sinh\left((N-1)\xi)\right)}-\frac{\sqrt{\nu\cosh^2\tfrac{\xi}{2}+\tfrac{1}{4}(\nu-1)^2\cosh^2\left(\left(N-\tfrac{1}{2}\right)\xi\right)}}{2\nu\sinh\left((N-1)\xi)\right)}\right]=\frac{e^{-\frac{\xi}{2}}}{2},
     \\
      &\lim_{\xi\rightarrow0^-}\left[-\frac{1+\nu}{4\nu}\frac{\cosh\left(\left(N-\tfrac{1}{2}\right)\xi\right)}{\sinh\left((N-1)\xi)\right)}-\frac{\sqrt{\nu\cosh^2\tfrac{\xi}{2}+\tfrac{1}{4}(\nu-1)^2\cosh^2\left(\left(N-\tfrac{1}{2}\right)\xi\right)}}{2\nu\sinh\left((N-1)\xi)\right)}\right]=-\frac{1+\nu}{2\nu(N-1)\xi}.
 \end{align}
 
  Let us now turn our attention to the value of $\alpha$, such that, the system transitions from the underdamped case (only one root in the imaginary axis) to the overdamped one (3 roots in the imaginary axis). The right hand side of Eq.~(\ref{alpha2}) has one minimum for $\xi<0$. This minimum defines the critical value $\alpha_c$ in which the system posses one root with multiplicity 2 in the imaginary axis. In fact, $\alpha_c$ cannot be expressed in some explicit formula, however, from the inequality
 \begin{align}
     \frac{e^{-\tfrac{1}{2}\xi}}{2}<-\frac{1+\nu}{4\nu}\frac{\cosh\left(\left(N-\tfrac{1}{2}\right)\xi\right)}{\sinh\left((N-1)\xi)\right)}-\frac{\sqrt{\nu\cosh^2\tfrac{\xi}{2}+\tfrac{1}{4}(\nu-1)^2\cosh^2\left(\left(N-\tfrac{1}{2}\right)\xi\right)}}{2\nu\sinh\left((N-1)\xi)\right)}< \frac{e^{-\tfrac{1}{2}\xi}}{2}\left[1-\frac{1+\nu^{-1}}{(N-1)\xi}\right],
 \end{align}
 we can show that
 \begin{equation}
     \frac{1}{2}<\alpha_c<\frac{1}{2}\exp\left[\frac{1+\nu}{4\nu(N-1)}\left(\sqrt{1+\frac{8\nu(N-1)}{1+\nu}}-1\right)\right]\frac{\sqrt{1+\frac{8\nu(N-1)}{1+\nu}}+1}{\sqrt{1+\frac{8\nu(N-1)}{1+\nu}}-1}\,.
 \end{equation}
 
\begin{figure}
\centering
\includegraphics[scale=0.8]{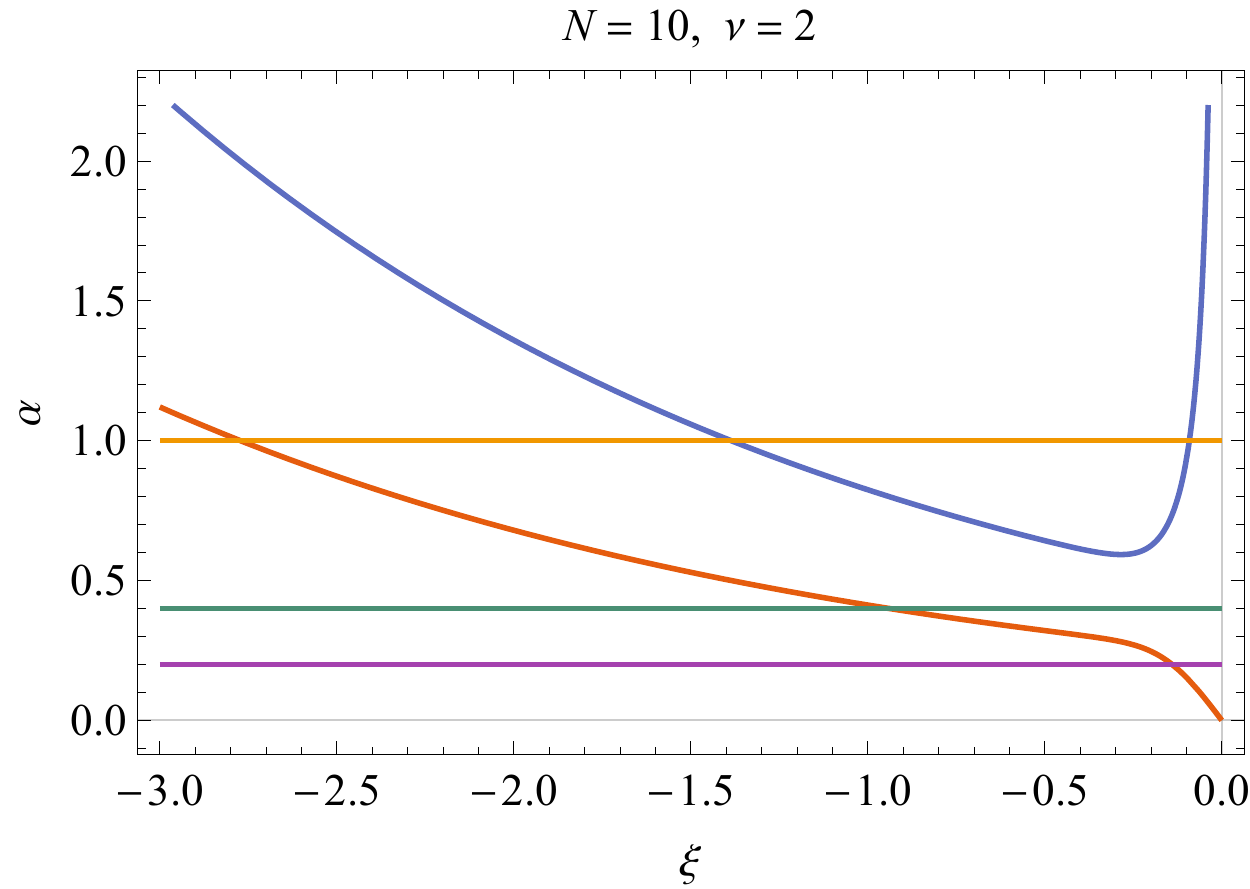}
\caption{Nomal mode frequencies on the imaginary axis for $\nu=2$ and $N=10$.}
\label{Normal mode2}
\end{figure}

 \subsection{Symmetric case $(\alpha_L=\alpha_R)$}
 
 We have shown in the main text that, for the particular case of $\alpha_L=\alpha_R=\alpha$,  Eq.~(\ref{D-alpha}) can be broken up into the two complex transcendental equations (\ref{alpha-tan}, \ref{alpha-cot}). Reexpressing Eq.~(\ref{alpha-tan}) into its real and imaginary part, we end up with
 \begin{align}
     \frac{\cos\left(\tfrac{\varphi}{2}\right)\sinh\left(\left(2N-1\right)\tfrac{\xi}{2}-\cos\left(\left(2N-1\right)\tfrac{\varphi}{2}\right)\sinh\left(\tfrac{\xi}{2}\right)\right)}{2\left[\cos\left((N-1)\varphi\right)+\cosh\left((N-1)\xi\right)\right]}&=\alpha\,,
     \\
     \sin\left(\left(2N-1\right)\tfrac{\varphi}{2}\right)\cosh\left(\tfrac{\xi}{2}\right)-\sin\left(\tfrac{\varphi}{2}\right)\cosh\left(\left(2N-1\right)\tfrac{\xi}{2}\right)&=0\,. \label{im-alpha-tan}
 \end{align}

 Equivalently, for Eq.~(\ref{alpha-cot}), we have
 \begin{align}
     \frac{\cos\left(\tfrac{\varphi}{2}\right)\sinh\left(\left(2N-1\right)\tfrac{\xi}{2}+\cos\left(\left(2N-1\right)\tfrac{\varphi}{2}\right)\sinh\left(\tfrac{\xi}{2}\right)\right)}{2\left[\cosh\left((N-1)\xi\right)-\cos\left((N-1)\varphi\right)\right]}&=\alpha\,,
     \\
     \sin\left(\left(2N-1\right)\tfrac{\varphi}{2}\right)\cosh\left(\tfrac{\xi}{2}\right)+\sin\left(\tfrac{\varphi}{2}\right)\cosh\left(\left(2N-1\right)\tfrac{\xi}{2}\right)&=0\,. \label{im-alpha-cot}
 \end{align}
 
 This means that the phonon frequencies lie either on the curve (\ref{im-alpha-tan}) or on the curve (\ref{im-alpha-cot}). From that, we are able to determine the lower bound for the normal mode frequencies, given in Eq.~(\ref{region}).
  
\begin{thm}
For $\alpha_L=\alpha_R$, all the zeros of $D(\omega(\theta))/B(\omega(\theta))$ are contained inside the region 
\[
-\frac{1}{N-1}\ln\left(\left|\csc\tfrac{\varphi}{2}\right|+\sqrt{\csc
^2\tfrac{\varphi}{2}+1}\;\right)<\xi<0,
\]
for $\varphi\in[-\pi,\pi]$.
\end{thm}
 
 \begin{proof}
   Multiplying Eqs.~(\ref{im-alpha-tan}) and (\ref{im-alpha-cot}), we find that the roots of $D(\omega(\theta))/B(\omega(\theta))$, when $\alpha_L=\alpha_R$ lie on the curve
  \begin{equation}
      \sin^2\left(\left(2N-1\right)\tfrac{\varphi}{2}\right)\cosh^2\left(\tfrac{\xi}{2}\right)-\sin^2\left(\tfrac{\varphi}{2}\right)\cosh^2\left(\left(2N-1\right)\tfrac{\xi}{2}\right)=0\,. \label{curve}
  \end{equation}
  Dividing this equation by $\cosh^2\left(\tfrac{\xi}{2}\right)\sin^2\left(\tfrac{\varphi}{2}\right)$, we find
  \begin{equation}
      0=\frac{\sin^2\left(\left(2N-1\right)\tfrac{\varphi}{2}\right)}{\sin^2\left(\tfrac{\varphi}{2}\right)}-\left[\cosh((N-1)\xi)+\tanh\left(\tfrac{\xi}{2}\right)\sinh((N-1)\xi)\right]^2\leq \frac{1}{\sin^2\left(\tfrac{\varphi}{2}\right)}-\cosh^2(N\xi), \label{ineq}
  \end{equation}
  If we restrict the analysis for $\xi\neq0$, we can replace the sign $\leq$ by $<$ in Eq.~(\ref{ineq}). This shows that the all the normal mode frequencies must be inside the region
  \begin{equation}
    \cosh((N-1)\xi)<\left|\csc\left(\tfrac{\varphi}{2}\right)\right|,  
  \end{equation}
  or equivalently
  \begin{equation}
      -\frac{1}{N-1}\ln\left(\left|\csc\tfrac{\varphi}{2}\right|+\sqrt{\csc
^2\tfrac{\varphi}{2}+1}\;\right)<\xi<\frac{1}{N-1}\ln\left(\left|\csc\tfrac{\varphi}{2}\right|+\sqrt{\csc
^2\tfrac{\varphi}{2}+1}\;\right)\,. \label{ineq2}
  \end{equation}
  
  \begin{figure}
\centering
\includegraphics[scale=0.64]{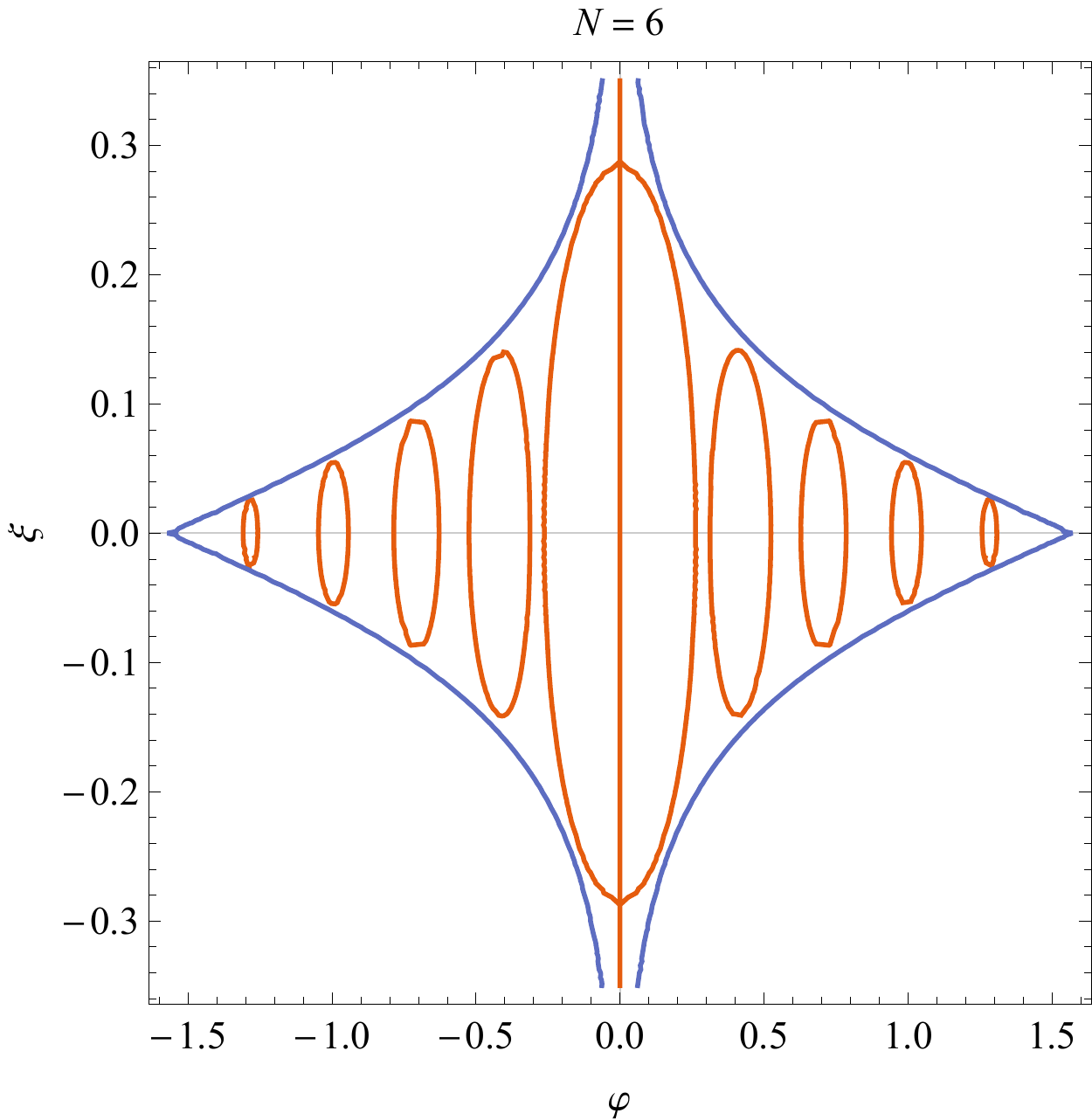}
\caption{In red is the curve from Eq.~(\ref{region}) and in blue is the curve $ \cosh((N-1)\xi)=\left|\csc\left(\tfrac{\varphi}{2}\right)\right|$.}
\label{Pole distribution}
\end{figure}
  We conclude this proof by combining Eq.~(\ref{ineq2}) with Theorem~\ref{thm1}.
 \end{proof}

\section{Heat Flux}

As shown in Eq.~(\ref{phi1}), the heat flux can be obtained solely through the asymptotic behavior of $\dot X_1(t)$ and $F_L(t)$. From Eq.~(\ref{X1-omega}), we can express $\dot X_1(t)$ as
\begin{equation}
 \dot X_1(t)=\int\limits_{-\infty}^{\infty}\frac{d\omega}{2\pi}(-i\omega)\tilde{X}_1(\omega)e^{-i\omega t}=\int\limits_{-\infty}^{\infty}\frac{d\omega}{2\pi}\int\limits_{-\infty}^{\infty}dt'\,\frac{e^{-i\omega(t-t')}}{im\omega_0D(\omega)}\left[A(\omega)F_L(t')+B(\omega)F_R(t')\right], \label{X1-dot}
\end{equation}
where $B(\omega)$, $A(\omega)$ and $D(\omega)$ are given by Eqs.~(\ref{B}-\ref{D}) respectively. Plugging Eq.~(\ref{X1-dot}) into Eq.~(\ref{phi1}), we end up with
\begin{align}
    J_E=&\lim_{t\rightarrow\infty}\int\limits_{-\infty}^{\infty}\frac{d\omega}{2\pi}\int\limits_{-\infty}^{\infty} \frac{d\omega'}{2\pi} \int\limits_{-\infty}^{\infty}dt'\int\limits_{-\infty}^{\infty}dt''\,\frac{e^{i\omega(t'-t)}\,e^{i\omega'(t''-t)}}{2m\omega_0D(\omega)D(\omega')}\left[A(\omega)\Big(2\alpha_LA(\omega')-iD(\omega')\Big)\left\langle \{\tilde F_{L}(t''),\tilde F_{L}(t')\}\right\rangle\right.\nonumber
    \\
    &+2\alpha_LB(\omega)B(\omega')\left\langle \{\tilde F_{R}(t''),\tilde F_{R}(t')\}\right\rangle\Bigg]\,. \label{J-interm}
\end{align}
Using Eq.~(\ref{Noise}), one can show that Eq.~(\ref{J-interm}) reduces to
\begin{align}
    J_E&=\int\limits_{-\infty}^{\infty} \frac{d\omega}{2\pi} \frac{\hbar\omega \alpha_L}{D(\omega)D(-\omega)} \left\{ \coth\left[\frac{\hbar\omega}{2k_BT_R}\right]  2\alpha_R B(\omega)B(-\omega)+\coth\left[\frac{\hbar\omega}{2k_BT_L}\right]A(-\omega) \left[ 2\alpha_LA(\omega)-iD(\omega) \right]  \right\},\nonumber
    \\
   \frac{ J_E}{\hbar\alpha_L\alpha_R}&=\int\limits_{-\infty}^{\infty} \frac{d\omega}{\pi} \frac{\omega B(\omega)B(-\omega)}{D(\omega)D(-\omega)} \left\{ \coth\left[\frac{\hbar\omega}{2k_BT_R}\right]   +\coth\left[\frac{\hbar\omega}{2k_BT_L}\right]\frac{4\alpha_LA(\omega)A(-\omega)-iD(\omega)A(-\omega)+iD(-\omega)A(\omega)}{4\alpha_RB(\omega)B(-\omega)}     \right\}. \label{JE-prelim}
\end{align}

Since $A(\omega)$, $B(\omega)$ and $D(\omega)$ are implicit functions of $\omega$, it is convenient to work explicitly with $\theta$. Because $\omega=2\omega_0\sin\tfrac{\theta}{2}$ is an odd function in $\theta$, we have that $\theta\rightarrow-\theta$ corresponds to $\omega\rightarrow-\omega$. Using this property, we find that
\begin{equation}
   \frac{4\alpha_LA\big(\omega)A(-\omega)-iD(\omega)A(-\omega)+iD(-\omega)A(\omega)}{4\alpha_RB(\omega)B(-\omega)}  =-2\,\frac{\cos^2\left((N-\tfrac{1}{2})\theta\right)+\sin\left(N\theta\right)\sin\left((N-1)\theta\right)}{\cos^2\tfrac{\theta}{2}}=-1,
\end{equation}
and we end up with Eq.~(\ref{heat-flux}), that is,
\begin{equation*}
    J_E= \frac{\hbar\alpha_L\alpha_R}{\pi}\int\limits_{-\infty}^{\infty}d\omega \frac{\omega B(\omega)B(-\omega)}{D(\omega)D(-\omega)} \left[ \coth\left(\frac{\hbar\omega}{2k_BT_R}\right)   -\coth\left(\frac{\hbar\omega}{2k_BT_L}\right)     \right].
\end{equation*}
\end{document}